\documentclass[pra,aps,twocolumn,nofootinbib, a4paper,superscriptaddress,floatfix]{revtex4-1}

\usepackage[T1]{fontenc}
\usepackage{times}
\usepackage{graphicx}
\usepackage{epsfig}
\usepackage{amsmath,amsthm}
\usepackage{bm}
\usepackage[colorlinks=true,citecolor=blue,urlcolor=blue]{hyperref}
\usepackage{braket}
\usepackage{color}
\newcommand{\be}{\begin{equation}}
\newcommand{\ee}{\end{equation}}
\newcommand{\ba}{\begin{eqnarray}}
\newcommand{\ea}{\end{eqnarray}}

\newcommand{\tr}{\operatorname{tr}}
\newcommand{\proj}[1]{\ket{#1}\!\!\bra{#1}}
\newcommand{\one}{{\bf{1}}}

\newcommand{\vecP}{\vec{P}}
\newcommand{\N}{\mathcal{N}}
\newcommand{\Q}{\mathcal{Q}}

\newtheorem{cor}{Corollary}
\newtheorem{observation}{Observation}

\newtheorem{proposition}{Proposition}

\renewcommand{\L}{\mathcal{L}}

\begin{document}

\title{Maximal Violation of a Broad Class of Bell Inequalities and Its Implications on Self-Testing}

\author{C. Jebarathinam}
\affiliation{Department of Physics and Center for Quantum Frontiers of Research \& Technology (QFort), National Cheng Kung University, Tainan 701, Taiwan}
\author{Jui-Chen Hung}
\affiliation{Department of Physics, National Cheng Kung University, Tainan 701, Taiwan}
\author{Shin-Liang Chen}
\affiliation{Department of Physics and Center for Quantum Frontiers of Research \& Technology (QFort), National Cheng Kung University, Tainan 701, Taiwan}
\affiliation{Dahlem Center for Complex Quantum Systems, Freie Universit\"at Berlin, 14195 Berlin, Germany}
\author{Yeong-Cherng Liang}
\email{ycliang@mail.ncku.edu.tw}
\affiliation{Department of Physics and Center for Quantum Frontiers of Research \& Technology (QFort), National Cheng Kung University, Tainan 701, Taiwan}
\date{\today}

\begin{abstract}
In quantum information, lifting is a systematic procedure that can be used to derive---when provided with a seed Bell inequality---other Bell inequalities applicable in more complicated Bell scenarios. It is known that the procedure of lifting introduced by Pironio [J. Math. Phys. A 46, 062112 (2005)] preserves the facet-defining property of a Bell inequality. Lifted Bell inequalities therefore represent a broad class of Bell inequalities that can be found in {\em all} Bell scenarios. Here, we show that the maximal value of {\em any} lifted Bell inequality is preserved for both the set of nonsignaling correlations and quantum correlations. Despite the degeneracy in the maximizers of such inequalities, we show that the ability to self-test a quantum state is preserved under these lifting operations. In addition, except for outcome-lifting, local measurements that are self-testable using the original Bell inequality---despite the degeneracy---can also be self-tested using {\em any} lifted Bell inequality derived therefrom. While it is not  possible to self-test {\em all} the positive-operator-valued measure elements using an outcome-lifted Bell inequality, we show that partial, but robust self-testing statements on the underlying measurements can nonetheless be made from the quantum violation of these lifted inequalities. We also highlight the implication of our observations on the usefulness of using lifted Bell-like inequalities as a device-independent witnesses for entanglement depth. The impact of the aforementioned degeneracy on the geometry of the quantum set of correlations is briefly discussed.
\end{abstract}

\maketitle

\section{Introduction}
Inspired by the thought-provoking paper of Einstein, Podolsky, and Rosen~\cite{EPR35}, Bell derived~\cite{Bell64}---based on well-accepted classical intuitions---an inequality constraining the correlations between local measurement outcomes on two distant systems.  He further showed that the so-called Bell inequality can be violated by quantum theory using local, but incompatible measurements on entangled states. Since then, various Bell inequalities, such as the one due to Clauser, Horne, Shimony and Holt (CHSH)~\cite{Clauser69} have been derived to investigate the intriguing nature of quantum theory and also the information-processing power enabled by these nonclassical, Bell-nonlocal~\cite{Brunner14} correlations.

For example, Ekert~\cite{Eke91} showed in 1991 that the quantum violation of Bell inequalities offers an unprecedented level of security for quantum key distribution protocols.  Independently, Mayers and Yao~\cite{Mayers:1998aa, MY04} showed that certain extremal quantum correlation enables the possibility to {\em self-test} quantum devices. These discoveries prompted the paradigm of device-independent quantum information~\cite{Scarani12,Brunner14} in which the physical properties of quantum devices are certified without making any assumption on the internal working of the devices, but rather through the observation of a Bell-nonlocal correlation.

Interestingly, an observation of the {\em maximal}  quantum violation of certain Bell inequalities, such as the CHSH inequality, is already sufficient to self-test the underlying quantum state and the measurements that give rise to the observed violation~\cite{PR92}. To date, numerous Bell inequalities have been derived (see, e.g.,~\cite{Braunstein1990,WW01,Collins2002,Kaszlikowski2002,Zukowski2002,Sliwa2003,Collins04,Buhrman2005,Guhne2005,Avis2006,Liang2009,Gisin:2009aa,Acin2012,Bancal12,Grandjean2012,Brunner14,Mancinska:2014aa,Liang:PRL:2015,Schwarz16,SAT+17,Enky1810.05636,Cope18,BAS19,Zambrini19} and references therein). However,  beyond the CHSH Bell inequality, only a handful of them~\cite{Miller2013,YN13,Mancinska:2014aa,Bamps2015,SAS+16,Jed16,Andersson2017,SAT+17,KST18,BAS19} have been identified as relevant for the purpose of self-testing (see~\cite{SB19} for a recent review on the topic of self-testing). Is it possible to make some general statements regarding the self-testing property of Bell inequalities defined for an {\em arbitrary} Bell scenario?

To answer the above question, we consider, in this work, Bell inequalities that may be obtained from the procedure of Pironio's {\em lifting}~\cite{PIR05}. Importantly, such inequalities exist in {\em all} Bell scenarios beyond the simplest one for two parties, two settings and two outcomes. If a Bell inequality is facet-defining~\cite{Collins04}, the same holds for its liftings~\cite{PIR05}. What about their quantum violation? In \cite{Rosset2014} (see also~\cite{Curchod2015}), it was shown that in addition to the bound satisfied by Bell-local correlations, the maximal quantum (nonsignaling~\cite{BLM+05}) value of Bell inequalities is preserved for party-lifting. In this work, we give an alternative proof of this fact and show, in addition, that the maximal quantum (non-signaling) value of a Bell inequality is also preserved for {\em other} types of Pironio's lifting.

As a corollary of our results, we further show that the self-testing properties of a Bell inequality is largely preserved through the procedure of lifting. In other words, if a Bell inequality can be used to self-test some quantum state $\ket{\psi}$, so can its liftings. Moreover, except for the case of outcome-lifting, the possibility to self-test the underlying measurement operators using a Bell inequality remains intact upon the application of Pironio's lifting. As we illustrate in this work, the maximizers of lifted Bell inequalities are not unique. There is thus no hope (see, e.g.,~\cite{GKW+18}) of providing a {\em complete} self-testing of the employed quantum device using a lifted Bell inequality. Nonetheless, we provide numerical evidence suggesting that lifted Bell inequalities provide the same level of robustness for self-testing the relevant parts of the devices.

The rest of this paper is organized as follows. In Sec. \ref{prl}, we introduce basic notions of a Bell scenario and remind the definitions of self-testing.
After that, we investigate and compare the maximal violation of lifted Bell inequalities against that of the original Bell inequalities, assuming quantum, or general nonsignaling correlations \cite{BLM+05}. In the same section, we also discuss the self-testing property of lifted Bell inequalities, and the usefulness of party-lifted Bell-like inequalities as device-independent witnesses for entanglement depth~\cite{Liang:PRL:2015}. In Sec.~\ref{con}, we present some concluding remarks and possibilities for future research. Examples illustrating the non-uniqueness of the maximizers of lifted Bell inequalities, as well as their implications on the geometry of the quantum set of correlations are provided in the Appendices.

\section{Preliminaries}\label{prl}
\subsection{Bell scenario}

Consider a Bell scenario involving $n$ spatially separated parties labeled by $i \in \{1,2,\cdots,n\}$ and
let the $i$-th party performs a measurement labeled by $j_i$,with outcomes labeled by $k_i$.
In any given Bell scenario, we may appreciate the strength of correlation between measurement outcomes  via a collection of joint conditional probabilities.
Following the literature (see, e.g.,~\cite{Brunner14}), we represent these conditional probabilities---dubbed a {\em correlation}---of getting the outcome combination $\vec{k}=(k_1k_2\cdots k_n)$ conditioned on performing the measurements $\vec{j}=(j_1j_2\cdots j_n)$ by the vector 
$\vecP:=\{P(\vec{k}|\vec{j})\}$. In addition to the normalization condition and the positivity constraint $P(\vec{k}|\vec{j})\ge 0$ for all $\vec{k},\vec{j}$,
each correlation is required to satisfy the nonsignaling constraints~\cite{BLM+05} (see also~\cite{PR94})
\begin{align}\label{Eq:NS}
&\sum_{k_i} P(k_1 \cdots k_i \cdots k_n|j_1 \cdots j_i \cdots j_n) \nonumber \\ 
&=\sum_{k_i} P(k_1 \cdots k_i \cdots k_n|j_1 \cdots j'_i \cdots j_n) 
\end{align}
for all $i, j_1, \cdots j_{i-1},j_i,j'_i,j_{i+1} \cdots  j_n$, and $k_\ell$ (with $\ell\neq i$).
In any given Bell scenario, the set of correlations satisfying these constraints forms the so-called nonsignaling 
polytope $\mathcal{N}$ \cite{BLM+05}.

A correlation  is called Bell-local~\cite{Brunner14} if it can be explained by a local-hidden-variable model~\cite{Bell64}, 
\begin{equation*}
P(k_1 \cdots k_n|j_1 \cdots j_n)=\sum_\lambda p_\lambda P(k_1|j_1, \lambda)\cdots P(k_n|j_n, \lambda)
\end{equation*}
for all $k_1 \cdots k_n,j_1 \cdots j_n$, where $\lambda$ is the hidden variable which occurs with probability $p_\lambda$,
$\sum_\lambda p_\lambda=1$ and $P(k_i|j_i, \lambda)$ is the probability of obtaining the measurement outcome 
$k_i$ given the setting $j_i$ and the hidden variable $\lambda$. In a given Bell scenario, the set of Bell-local
correlations forms a polytope  called a Bell polytope, or more frequently a local polytope $\mathcal{L}$, which is a subset of the nonsignaling polytope.

A correlation which cannot be explained by a local-hidden-variable model is said to be Bell-nonlocal 
and must necessarily violate a Bell inequality~\cite{Bell64} --- a constraint satisfied by all $\vecP\in\L$.
A linear Bell inequality has the generic form:
\begin{align}\label{Eq:npartiteBI}
	I_n(\vec{P}):=\vec{B}\cdot \vec{P}=\sum_{k_1\cdots k_n,j_1\cdots j_n}  B_{\vec{k},\vec{j}}P(\vec{k}|\vec{j})
\stackrel{\L}{\le} \beta_{\L},
\end{align}
where $\vec{B}:=\{B_{\vec{k},\vec{j}}\}$ denotes the vector of Bell coefficients, $\vec{B}\cdot \vec{P}$ and $\beta_{\L}$ are, respectively, the Bell expression and the local bound of a Bell inequality.

A correlation $\vecP$ is called quantum if the joint probabilities  can be written as
\be\label{Eq:Born1}
	P(k_1 \cdots k_n|j_1 \cdots j_n)=\tr \left( \otimes^n_{i=1}  M^{(i)}_{k_i|j_i} \rho_{12\cdots n} \right) 
\ee
where $\rho_{12\cdots n}$ is an $n$-partite density matrix and $\{M^{(i)}_{k_i|j_i}\}_{k_i}$ is the positive-operator-valued measure (POVM) describing the $j_i$-th measurement of the $i$-party.  By definition, POVM elements satisfy the constraints of being positive semidefinite, $M^{(i)}_{k_i|j_i}\succeq 0$ for all $k_i$ and $j_i$, as well as the normalization requirement $\sum_{k_i} M^{(i)}_{k_i|j_i}= \one$ for all $j_i$. Thus, a correlation is quantum if and only if the joint probabilities of such a correlation can be realized experimentally by performing local measurements on an $n$-partite quantum system.  The set of quantum correlations $\Q$  forms a convex set satisfying $\mathcal{L} \subset \mathcal{Q}\subset \mathcal{N}$. It is, however, not a polytope \cite{WW01} (see also~\cite{GKW+18}). When necessary, we will use $\Q_n$ to denote the set of quantum correlations arising in an $n$-partite Bell scenario.

\subsection{Self-testing}
\label{Sec:SelfTest}

Certain nonlocal correlations $\vecP\in\Q$ have the appealing feature of being able to reveal (essentially unambiguously) the quantum strategy, i.e., the underlying  state and/or the measurement operators leading to these correlations~\cite{PR92,WW01,MY04,SB19}. Following  \cite{MY04}, we say that such a $\vecP\in\Q$ self-tests the underlying quantum strategy. To this end, it is worth noting that all pure bipartite entangled states can be self-tested~\cite{Coladangelo:2017aa}.

To facilitate subsequent discussions, we remind here the formal definition of self-testing in a bipartite Bell scenario following the approach of \cite{GKW+18} (see also~\cite{SB19}). Specifically, consider two spatially separated parties Alice and Bob who each performs measurements labeled by $x$, $y$ and, respectively, observes the outcomes $a$, $b$. 
We say that a bipartite correlation $\vec{P}:=\{P(ab|xy)\}$ satisfying 
\begin{equation}\label{Eq:Born2}
	P(ab|xy)=\tr\left[ M^{(1)}_{a|x} \otimes M^{(2)}_{b|y}  \rho_{12}\right], 
\end{equation}
for all  $a$, $b$, $x$ and $y$ self-tests the reference (entangled) state $\ket{\tilde{\psi}_{12}}$ if there exists a local isometry $\Phi=\Phi_1 \otimes  \Phi_2$ such that
\be \label{Eq:StateTransformation}
	\Phi\, \rho_{12}\, \Phi^\dag=   \proj{\tilde{\psi}_{12}} \otimes \rho_{aux}
\ee
where $\rho_{12}$ is the measured quantum state [acting on $\mathcal{H}_A \otimes \mathcal{H}_B$],  $\rho_{aux}$ is an auxiliary state acting 
on  $\mathcal{H}_{A'} \otimes \mathcal{H}_{B'}$, and $\mathcal{H}_{A'}$ and $\mathcal{H}_{B'}$ are the Hilbert spaces associated with the other degrees of freedom of Alice and Bob's subsystem respectively~\cite{CKS19}. 

Often, a $\vecP\in\Q$  that self-tests some reference quantum state can also be used to certify the measurements as well. In such cases, 
we say that a bipartite correlation $\vecP$ obtained from Eq.~\eqref{Eq:Born2} self-tests the reference quantum state $\ket{\tilde{\psi}_{12}}$ and the reference POVM
$\{\tilde{M}^{(1)}_{a|x}\}_{a}$, $\{\tilde{M}^{(2)}_{b|y}\}_{b}$ if  there exists a local isometry
$\Phi=\Phi_1  \otimes \Phi_2$ such that Eq.~\eqref{Eq:StateTransformation} holds and 
\begin{align}\label{Eq:POVMTransformation}
	&\Phi [ M^{(1)}_{a|x} \otimes M^{(2)}_{b|y}  \rho_{12} ] \Phi^\dag
	= [ \tilde{M}^{(1)}_{a|x}  \otimes \tilde{M}^{(2)}_{b|y} \proj{\tilde{\psi}_{12}} ] \otimes \rho_{aux}. 
\end{align}
for all  $a$, $b$, $x$ and $y$. By summing over $a,b$, and using the normalization of POVM, one recovers Eq.~\eqref{Eq:StateTransformation} from Eq.~\eqref{Eq:POVMTransformation}.

Interestingly, there are Bell inequalities whose maximal quantum violation alone is sufficient to self-test the 
quantum state (and the measurement operators) \cite{MYS12,YN13,PVN14,SAS+16,Jed16,SAT+17,Natarajan:2017,Coladangelo2018}. Since then, identifying Bell
inequalities which can be used for the task of self-testing has received considerable attention.
To this end, note that if the maximal quantum violation of a Bell inequality self-tests some quantum state as well as the underlying measurements, then this maximal quantum violation must be achieved by a unique correlation \cite{GKW+18}. 

More formally, consider a bipartite Bell inequality, 
\be \label{2partiteBI}
I_2(\vec{P}):=\sum_{a, b,x,y}  B_{a, b,x,y}P(ab|xy) \stackrel{\L}{\le} \beta_\L,
\ee
with a quantum bound (maximal quantum violation):
\begin{equation}
	\beta_\Q=\max_{\vec{P}\in\Q} I_2(\vec{P})>\beta_\L.
\end{equation}
We say that an observation of the quantum violation $I_2(\vec{P})=\beta_\Q$ self-tests the reference (entangled) state $\ket{\tilde{\psi}_{12}}$ and the reference POVM $\{\tilde{M}^{(1)}_{a|x}\}_{a}$, $\{\tilde{M}^{(2)}_{b|y}\}_{b}$ if there exists a local isometry
$\Phi=\Phi_1 \otimes \Phi_2$ such that Eq.~\eqref{Eq:StateTransformation}, Eq.~\eqref{Eq:POVMTransformation}  hold for all $\vec{P}\in\Q$ satisfying Eq.~\eqref{Eq:Born2} and $I_2(\vec{P})=\beta_\Q$.

\section{Maximal violation of lifted Bell inequalities and its implications}
\label{maxlift}

In this section, we show that the maximal quantum (nonsignaling) violation of a lifted Bell inequality must be the same as that of the Bell inequality from which the lifting is applied.  We then discuss the implication of these observations in the context of self-testing, on the geometry of the quantum set of correlations, as well as on the device-independent certification of entanglement depth. For ease of presentation, our discussion will often be carried out assuming a bipartite Bell scenario (for the original Bell inequality). However, it should be obvious from the presentation that our results also hold for any Bell scenario with more parties, and also for a Bell inequality that is not necessarily facet-defining.

\subsection{More inputs}

Let us begin with the simplest kind of lifting, namely, one that allows additional measurement settings. Applying Pironio's input-lifting~\cite{PIR05} to a Bell inequality means to consider the very same Bell inequality in a Bell scenario with more measurement settings for at least one of the parties.  At first glance, it may seem rather unusual to make use of only the data collected for a subset of the input combinations, but in certain cases (see, e.g.,~\cite{Shadbolt2012}), the consideration of all input-lifted facets is already sufficient to identify the non-Bell-local nature of the observed correlations.

Since an input-lifted Bell inequality is exactly the same as the original Bell inequality, its maximal quantum (non-signaling) violation is obviously the same as that of the original Bell inequality. Similarly, it is evident that if the maximal quantum violation of the original Bell inequality self-tests some reference quantum state $\ket{\psi}$ and POVMs $\{M^{(1)}_{a|x}\}_{a,x}$,$\{M^{(2)}_{b|y}\}_{b,y}$,\ldots, so does the maximal quantum violation of the input-lifted Bell inequality. 

However, since no constraint is imposed on the additional inputs that do not appear in the Bell expression, it is clear that even if we impose the constraint that the maximal quantum violation of an input-lifted Bell inequality is attained, these other local POVMs can be completely arbitrary. Thus, the subset of quantum correlation attaining the maximal quantum violation of {\em any} input-lifted Bell inequality is {\em not} unique, and has a degeneracy that increases with the number of these ``free" inputs. In other words, the set of quantum maximizers of any input-lifted Bell inequality define a flat region of the boundary of the quantum set of correlations, cf.~\cite{GKW+18}. In particular, it could lead to completely flat boundaries of $\Q$ on specific two-dimensional slices in the correlation space (see Fig.~\ref{Fig1}). For some explicit examples illustrating the aforementioned non-uniqueness, see Appendix~\ref{Examples}.

\subsection{More outcomes}

Instead of the trivial input-lifting, one may also lift a Bell inequality to a scenario with more measurement outcomes. Specifically, consider  a bipartite Bell scenario where 
the $y'$-th measurement of Bob has $v\ge2$ possible outcomes. 
The simplest outcome-lifting \`a la Pironio~\cite{PIR05} then consists of two steps: (1) choose an outcome, say, $b=b'$ from Bob's $y'$-th measurement, and (2) replaces in the sum of Eq.~\eqref{2partiteBI} all terms of the form $P(ab'|xy')$ by $P(ab'|xy')+P(au|xy')$.

The resulting outcome-lifted Bell inequality reads as:
\begin{align}\label{LO2partiteBI} 
I^\text{\tiny LO}_2=&\!\!\sum_{a, b,x,y\ne y'}  B_{a, b,x,y}P(ab|xy) +\!\!\sum_{a,x,b\ne b', u}  B_{a, b,x,y'}P(ab|xy') \nonumber  \\
&+\sum_{a,x}  B_{a, b',x,y'}\left[ P(ab'|xy') + P(au|xy')\right]\stackrel{\L}{\le} \beta_{L}    
\end{align}
where the local bound $\beta_\L$ is provably the same~\cite{PIR05} as that of the original Bell inequality, Eq.~\eqref{2partiteBI}. It is worth noting that outcome-lifted Bell inequalities arise naturally in the study of detection loopholes in Bell experiments, see, e.g.,~\cite{Branciard2011,Cope18}.

\subsubsection{Preservation of quantum and nonsignaling violation}

As with input-lifting, we now proceed to demonstrate the invariance of maximal Bell violation with outcome-lifting.

\begin{proposition}\label{Obs:MaxValuePreserved-LO}
Lifting of outcomes preserves the quantum bound and the nonsignaling  bound of any Bell inequality, i.e., $\beta^\text{\tiny LO}_{\Q}=\beta_{\Q}$ and $\beta^\text{\tiny LO}_{\N}=\beta_{\N}$, where $\beta_{\Q}$ ($\beta^\text{\tiny LO}_\Q$) and $\beta_{\N}$ ($\beta^\text{\tiny LO}_{\N}$) are, respectively, the quantum and the nonsignaling bounds  of the original (outcome-lifted) Bell inequality.
\end{proposition}
\begin{proof}
From Eq.~\eqref{LO2partiteBI}, one clearly sees that the $b'$-th outcome and the $u$-th outcome of Bob's $y'$-th measurement are treated on equal footing. So, we may as well consider Bob's $y'$-th measurement as an {\em effective} $v$-outcome measurement by considering its $b'$-th outcome and its $u$-th outcome together as one outcome. Hence, if we define
\be \label{coarse-grain}
\begin{split}
	\tilde{P}(ab|xy)&=P(ab|xy),\quad y\neq y',\\
	\tilde{P}(ab|xy')&=P(ab|xy'),\quad b\not\in\{b',u\},\\
	\tilde{P}(ab'|xy')&=P(ab'|xy') + P(au|xy')
\end{split}
\ee
and substitute it back into Eq.~\eqref{LO2partiteBI}, we recover the Bell expression of the original Bell inequality [left-hand-side of Eq.~\eqref{2partiteBI}] by identifying $\tilde{P}(ab|xy)$ in $I^\text{\tiny LO}_2$ as ${P}(ab|xy)$ in $I_2$. Moreover,  if $\vec{P}$ defined for this more-outcome Bell scenario is quantum (nonsignaling), the resulting correlation obtained with the coarse-graining procedure of Eq.~\eqref{coarse-grain} is still quantum (nonsignaling). A proof of this for quantum correlations is provided in Appendix~\ref{QRpre} (see, e.g., \cite{Jul14}  for the case of nonsignaling correlation).

This implies that for any violation of the outcome-lifted Bell inequality (\ref{LO2partiteBI}) by a quantum (nonsignaling) 
correlation, there always exists another quantum (nonsignaling) correlation that gives the same amount of violation for 
the original Bell inequality   (\ref{2partiteBI}). In particular, the maximal quantum (nonsignaling) violation of these inequalities must satisfy
\be \label{OtoL}
\beta_{\N} \ge \beta^\text{\tiny LO}_{\N} \quad \beta_{\Q} \ge \beta^\text{\tiny LO}_{\Q}.
\ee

On the other hand, instead of grouping the outcomes in the outcome-lifted Bell scenario, one could also start from the original Bell scenario and (arbitrarily) {\em split} the $b'$-th outcome of Bob's $y'$-th measurement into two outcomes labeled by $b=b'$ and $b=u$. Hence, if we define $\widehat{P}(ab'|xy')$, $\widehat{P}(au|xy')$ in the outcome-lifted Bell scenario such that 
\be \label{Eq:fine_grain}
\begin{split}
	\widehat{P}(ab|xy) = P(ab|xy),\quad &y\neq y'\text{\ or\ } y=y',b\neq b',u\\
	0\le \widehat{P}(ab'|xy')&,\,\, \widehat{P}(au|xy')\le 1,\\
	\widehat{P}(ab'|xy') + \widehat{P}&(au|xy')=P(ab'|xy')
\end{split}
\ee
and substitute it into Eq.~\eqref{2partiteBI}, we recover the outcome-lifted Bell expression [Eq.~\eqref{LO2partiteBI}] by identifying $\widehat{P}(ab'|xy')$ and $\widehat{P}(au|xy')$, respectively, as  $P(ab'|xy')$ and $P(au|xy')$ in $I^\text{\tiny LO}_2$, cf Eq.~\eqref{LO2partiteBI}. 

Moreover, the correlation obtained by locally splitting the outcomes, as required in Eq.~\eqref{Eq:fine_grain}, is realizable quantum-mechanically (see Appendix~\ref{QRpre}) or in general nonsignaling theory (see \cite{Jul14}) if the original correlations are, respectively, quantum and nonsignaling. Hence, for any violation of the original Bell inequality [Eq.~\eqref{2partiteBI}] by a quantum (nonsignaling) correlation, there always exists a quantum (nonsignaling)  correlation giving the same amount of violation for the outcome-lifted Bell inequality   (\ref{LO2partiteBI}), i.e., 
\be \label{LtoO}
\beta_{\Q} \le \beta^\text{\tiny LO}_{\Q}; \quad \beta_{\N} \le \beta^\text{\tiny LO}_{\N}.
\ee 
Combining Eqs. (\ref{OtoL}) and (\ref{LtoO}), it then follows that the maximal quantum (nonsignaling) violation of any Bell inequality is {\em preserved} through the procedure of outcome-lifting, i.e., 
\be 
\beta_{\Q} = \beta^\text{\tiny LO}_{\Q}; \quad \beta_{\N} = \beta^\text{\tiny LO}_{\N}.
\ee 
This completes the proof when only one of the outcomes ($b=u$) of one of the measurements ($y=y'$) of one of the parties (Bob) is lifted. However, since more complicated outcome-lifting can be achieved by concatenating the simplest outcome-lifting presented above, the proof for the general scenarios can also be obtained by concatenating the proof given above, thus completing the proof for the general scenario.
\end{proof}

\subsubsection{Implications on self-testing}

As an implication of the above Proposition, we obtain the following result in the context of quantum theory.
\begin{cor}\label{Res:SelfTestState}
If the maximal quantum violation of a Bell inequality self-tests a quantum  state $\ket{\tilde{\psi}}$, then
 any Bell inequality obtained therefrom by outcome-lifting also self-tests $\ket{\tilde{\psi}}$.
\end{cor}

\begin{proof}

For definiteness, we prove this for the specific case of $n=2$, the general proof is completely analogous. To this end, let $\rho^*_{12}$ denote an optimal quantum state that maximally violates  the outcome-lifted Bell inequality (\ref{LO2partiteBI}) with appropriate choice of POVMs $\{M^{(1)}_{a|x}\}_{a,x}$,$\{M^{(2)}_{b|y}\}_{b,y}$. 
As shown in Appendix \ref{QRpre}, this quantum state $\rho^*_{12}$ can also be used to realize an effective $v$-outcome distribution for Bob's $y'$-th measurement by combining Bob's relevant POVM elements for this measurement into a single POVM element, thereby implementing the local coarse graining given in Eq. (\ref{coarse-grain}) to give the maximal quantum violation of the original Bell inequality, Eq.~\eqref{2partiteBI}. Suppose that the maximal quantum violation of inequality~\eqref{LO2partiteBI} {\em does not} self-test the reference state $\ket{\tilde{\psi}_{12}}$, i.e., there {\em does not} exist {\em any} local isometry $\Phi=\Phi_1 \otimes \Phi_2$ such that 
 \be 
	\Phi \rho^*_{12} \Phi^\dag=  \proj{\tilde{\psi}_{12}} \otimes \rho_{aux} 
\ee 
for some $\rho_{aux}$. Then, we see that the maximal quantum violation of inequality~\eqref{2partiteBI} (attainable using $\rho^*_{12}$) also cannot self-test the reference state $\ket{\tilde{\psi}_{12}}$. The desired conclusion  follows by taking the contrapositive of the above implication.
\end{proof}

A few remarks are now in order. As with any other Bell inequality, in examining the quantum violation of an outcome-lifted Bell inequality, one may consider {\em arbitrary} local POVMs having the right number of outcomes (acting on some given Hilbert space). {\em A priori}, they do not have to be related to the optimal POVM of the original Bell inequality. However, from the proof of Proposition~\ref{Obs:MaxValuePreserved-LO}, one notices that POVMs arising from splitting the outcomes of the original optimal POVM do play an  important role in attaining the maximal quantum violation of the outcome-lifted Bell inequality. 

The arbitrariness in this splitting, nonetheless, implies that $\vecP\in\Q$ maximally violating an outcome-lifted Bell inequality are {\em not unique} (see Appendix~\ref{App:LO} for some explicit examples). Since this invalidates a necessary requirement to self-test both the state and {\em all} the local POVMs (see Proposition $C.1.$ of Ref. \cite{GKW+18}), we must thus conclude---given that such an inequality preserves the ability to self-test the underlying state---that its maximal violation cannot be used to {\em completely} self-test the underlying measurements. Using the swap method of~\cite{YVB+14}, we nevertheless show in Appendix~\ref{App:Self-test} that the quantum violation of an outcome-lifted Bell inequality may still provide robust self-testing of some of the underlying POVM elements, as well as the nature of the merged POVM elements.

\subsection{More parties}

Finally, let us consider the party-lifting of~\cite{PIR05}. Again, for simplicity, we provide hereafter explicit constructions and proofs only for the bipartite scenario, with the multipartite generalizations proceed analogously. To this end, it is expedient to write a generic bipartite Bell inequality such that
\be 
I_2:=\sum_{a,b,x,y}  B_{a,b,x,y}P(ab|xy)
\stackrel{\L}{\le} 0, \label{OnPBE}
\ee
i.e., with its local bound being zero.\footnote{This can always be achieved by (repeatedly) applying identity of the form given in Eq.~\eqref{Eq:NS} to both sides of Eq.~\eqref{2partiteBI}.} For any {\em fixed} but {\em arbitrary} input-output pair $c',z'$ of the additional party (Charlie), applying the party-lifting of~\cite{PIR05} to inequality~\eqref{OnPBE} gives rise to the tripartite Bell inequality:
\be 
I^\text{\tiny LP}_2:=\sum_{a,b,x,y}  B_{a,b,x,y}P(abc'|xyz')
\stackrel{\L}{\le} 0. \label{LPBE}
\ee
It is worth noting that such Bell inequalities has found applications in the foundations of quantum theory~\cite{Bancal:NatPhys:2012,Barnea2013}, as well as in the systematic generation~\cite{Curchod2015} of device-independent witnesses for entanglement depth~\cite{Liang:PRL:2015}.

\subsubsection{Preservation of quantum and nonsignaling violation}
\label{Sec:PartyLiftingPreservation}

That the maximal quantum and nonsignaling violation remain unchanged under Pironio's party-lifting operation~\cite{PIR05} follows directly from the results shown in Section 2.4 of~\cite{Rosset2014}, as well as a special case (with $n=k$) of Theorem 2 of~\cite{Curchod2015}. For the convenience of subsequent discussions, however, we provide below an alternative proof of this observation.

\begin{observation}\label{PartyLiftingPreservation}
Lifting of parties preserves the quantum bound and the nonsignaling  bound of any Bell inequality, i.e., $\beta^\text{\tiny LP}_{\Q}=\beta_{\Q}$ and $\beta^\text{\tiny LP}_{\N}=\beta_{\N}$, where $\beta_{\Q}$ ($\beta^\text{\tiny LP}_\Q$) and $\beta_{\N}$ ($\beta^\text{\tiny LP}_{\N}$) are, respectively, the quantum and the nonsignaling bounds  of the original (party-lifted) Bell inequality.
\end{observation}
\begin{proof}
For a tripartite Bell scenario relevant to inequality~\eqref{LPBE}, the marginal  probability of Charlie getting the outcome $c'$ conditioned on him performing the measurement labeled by $z'$ is:
\begin{equation}
	P(c'|z')=\sum_{a,b}P(abc'|xyz').
\end{equation}
Since the party-lifted inequality of Eq.~\eqref{LPBE} is saturated with the choice of $P(c'|z')=0$, thereby making $P(abc'|xyz')=0$ for all $a,b,x,y$, the observation holds trivially if  inequality~\eqref{LPBE} {\em cannot} be violated by general nonsignaling correlations.

Conversely, if inequality~\eqref{LPBE} can be violated by some quantum, or general nonsignaling correlation, the corresponding 
$P(c'|z')$ must be nonvanishing. Henceafter, we thus assume that $P(c'|z')>0$. 
To this end, note that 
\be\label{Eq:Dfn:ConditionalBipartite}
{P}_{c'|z'}(ab|xy):=P(abc'|xyz')/P(c'|z'),
\ee
gives the probabilities of Alice and Bob obtaining the outcomes $a$ and $b$ conditioned on her (him) choosing measurement $x$ ($y$),  Charlie
measuring $z'$ and obtaining the outcome $c'$. 
Note that the vector of probabilities  $ \vec{P}_{c'|z'}:=\{{P}_{c'|z'}(ab|xy)\}$ 
is a legitimate correlation in the (original) Bell scenario corresponding to inequality~(\ref{OnPBE}).

To prove the observation, we now focus on the case of finding the quantum bound, i.e., the maximum value of the left-hand-side of Eq.~\eqref{LPBE} for quantum correlations --- the proof for the nonsignaling case is completely analogous. To this end, note that the quantum bound of inequality~\eqref{LPBE}---given the above remarks---satisfies:
\begin{align}\label{Eq:BoundPartyLifting}
\beta^\text{\tiny LP}_{\Q}&=\max_{\{P(abc'|xyz')\} \in \Q_3} \sum_{a,b,x,y}  B_{a,b,x,y}P(abc'|xyz')\nonumber\\
&=\max_{\{P(abc'|xyz')\} \in \Q_3} \sum_{a,b,x,y}  B_{a,b,x,y}{P}_{c'|z'}(ab|xy)P(c'|z')\nonumber\\
&\le\max_{\vec{P}_{c'|z'} \in \Q_2} \sum_{a,b,x,y}  B_{a,b,x,y}{P}_{c'|z'}(ab|xy)\max P(c'|z')\nonumber\\
&=\max_{\vec{P}_{c'|z'} \in \Q_2} \sum_{a,b,x,y}  B_{a,b,x,y}{P}_{c'|z'}(ab|xy)= \beta_\Q,
\end{align}
where the first inequality follows from the fact that an independent maximization  over  $\vec{P}_{c'|z'} \in \Q_2$ and $P(c'|z')$ is, in principle, less constraining than a maximization over all tripartite quantum distributions $\{P(abc'|xyz')\}$, the second-last equality follows from the fact that $P(c'|z')\le1$ for legitimate marginal probability distributions, and the last equality follows from the fact that any bipartite quantum correlation can be seen as the marginalization of a tripartite one.

To complete the proof, note that the inequality $\beta^\text{\tiny LP}_{\Q}\le \beta_{\Q}$ can indeed be saturated if
\begin{subequations}
\begin{gather}
	P(abc'|xyz') = P^*(ab|xy)P(c'|z')\quad \forall\,\, a,b,x,y,\label{Eq:Factorization}\\
	P(c'|z') = 1,\label{Eq:Deterministic}\\
	\sum_{a,b,x,y}  B_{a,b,x,y}P^*(ab|xy) =  \beta_\Q.\label{Eq:Saturation}
\end{gather}
\end{subequations}
Moreover, these equations can be simultaneously satisfied with the three parties sharing a state of the form $\ket{\psi_{123}} = \ket{\psi^*_{12}}\otimes\ket{\psi_3}$ while
employing the local measurements:
\begin{equation}
	M^{(1)}_{a|x} = M^{(1*)}_{a|x},\quad M^{(2)}_{b|y} = M^{(2*)}_{b|y},\quad M^{(3)}_{c|z'} = \one\delta_{c,c'},
\end{equation}
where $\ket{\psi^*_{12}}, \{M^{(1*)}_{a|x}\}_{a,x}, \{M^{(2*)}_{b|y}\}_{b,y}$ consistute a maximizer for the (original) Bell inequality of Eq.~\eqref{OnPBE}, i.e., 
\begin{equation}
\begin{split}
	\beta_\Q&= \max_{\vec{P} \in \Q_2} \sum_{a,b,x,y}  B_{a,b,x,y}{P}(ab|xy)\\
	&= \sum_{a,b,x,y}  B_{a,b,x,y}{P^*}(ab|xy) \\
	&= \sum_{a,b,x,y}  B_{a,b,x,y}\bra{\psi^*_{12}} M^{(1*)}_{a|x}\otimes M^{(2*)}_{b|y}\ket{\psi^*_{12}}.
\end{split}
\end{equation}
\end{proof}

\subsubsection{Implications on self-testing}

As an implication of the above observation, we obtain the following result in the context of quantum theory.
\begin{cor}\label{Res:PartyLifting}
If a Bell inequality self-tests $\ket{\tilde{\psi}}$ and some reference POVMs $\{\tilde{M}^{(1)}_{a|x}\}_{a,x}, \{\tilde{M}^{(2)}_{b|y}\}_{b,y}$ etc., then the maximal quantum violation of  any Bell inequality obtained therefrom via party-lifting also self-tests the same state and the same local POVMs for an appropriate subset of parties.
\end{cor}
\begin{proof}
In the following, we use inequality~\eqref{LPBE} to illustrate how the proof works in the tripartite case. Note from Eq.~\eqref{Eq:BoundPartyLifting} that when the party-lifted Bell inequality of Eq.~(\ref{LPBE}) is violated to its quantum maximum $\beta_\Q$, the marginal distribution of Charlie necessarily satisfies $P(c'|z')=1$. It then follows from Eq.~\eqref{Eq:Dfn:ConditionalBipartite} that 
 \be\label{Eq:Factorization2}
 P(abc'|xyz')=P(ab|xy)P(c'|z') \quad \forall\,\, a,b,x,y.
 \ee
 Furthermore, from Eq.~\eqref{Eq:BoundPartyLifting}, this tripartite distribution gives the quantum bound of inequality~\eqref{LPBE} if and only if the marginal distributions $P(ab|xy)$ of Eq.~\eqref{Eq:Factorization2} violates the original Bell inequality of Eq.~\eqref{OnPBE} to its quantum bound.
 Therefore, if the original Bell inequality self-tests the $2$-partite entangled  state $\ket{\tilde{\psi}_{12}}$, then  for any tripartite density matrix $\rho_{123}$ leading to the quantum maximum of inequality Eq.~\eqref{LPBE}, there must exist a local isometry $\Phi=\Phi_1 \otimes \Phi_2 $ such that 
 \be
 \begin{split}
 &\Phi \tr_{3} \left(\rho_{123}\right)   \Phi^\dag=\proj{\tilde{\psi}}\otimes \rho_{aux},\\
 	&\Phi \left[ M^{(1)}_{a|x} \otimes M^{(2)}_{b|y}  \tr_{3} \left(\rho_{123}\right) \right] \Phi^\dag  \\
	= &\left( \tilde{M}^{(1)}_{a|x}  \otimes \tilde{M}^{(2)}_{b|y} \proj{\tilde{\psi}} \right) \otimes \rho_{aux}. 
 \end{split}
 \ee
 where $\rho_{aux}$ is some auxiliary density matrix acting on other degrees of freedom of Alice and Bob's subsystem. In other words, if the quantum maximum of the original Bell inequality can be used to self-test $\ket{\tilde{\psi}}$ and  reference POVMs $\{\tilde{M}^{(1)}_{a|x}\}_{a,x}$, $\{\tilde{M}^{(2)}_{b|y}\}_{b,y}$, so does the quantum maximum of the party-lifted Bell inequality.
\end{proof}

\subsubsection{Implications on device-independent certification of entanglement depth}

In Theorem 2 of~\cite{Curchod2015}, it was shown that  if
\begin{align}\label{Eq:size-k}
	\sum_{k_1\cdots k_n,j_1\cdots j_n}  B_{\vec{k},\vec{j}}P(\vec{k}|\vec{j})\stackrel{\mathcal{R}_{n,\ell}}{\le} 0,
\end{align}
is satisfied by an $n$-partite resource $\mathcal{R}$ (quantum or nonsignaling) that has a group size of $\ell$, its lifting to $(n+h)$ parties also holds for the same kind of resource of group size $\ell$:
\begin{align}\label{Eq:size-k:n+h}
	\sum_{k_1\cdots k_n,j_1\cdots j_n}  B_{\vec{k},\vec{j}}P(\vec{k},\vec{o}|\vec{j},\vec{s})\stackrel{\mathcal{R}_{n+h,\ell}}{\le} 0,
\end{align}
where $\vec{o}$ ($\vec{s}$) is any {\em fixed}, but {\em arbitrary} string of outputs (inputs) for the $h$ additional parties. For the case of $\ell=n$, the above result reduces to Observation~\ref{PartyLiftingPreservation} discussed in Sec.~\ref{Sec:PartyLiftingPreservation}.

When the considered resource is restricted to shared quantum correlations, inequality~\eqref{Eq:size-k} and inequality~\eqref{Eq:size-k:n+h} are instances of so-called 
device-independent witnesses for entanglement depth~\cite{Liang:PRL:2015}, i.e., Bell-like inequalities capable of certifying---directly from the observed correlation---a lower bound on the entanglement depth~\cite{Sorensen:PRL:2001} of the measured system. More specifically, if the observed quantum value of the left-hand-side of Eq.~\eqref{Eq:size-k} or Eq.~\eqref{Eq:size-k:n+h} is greater than 0, then one can certify that the locally measured quantum state must have an entanglement depth of at least $\ell+1$.

Although the above result of~\cite{Curchod2015} can be applied to an arbitrary number of $(n+h)$ parties, Observation~\ref{PartyLiftingPreservation} implies that if the seed inequality is that applicable to an $n$-partite Bell scenario, the extended scenario can never be used to certify an entanglement depth that is larger than $n$. This follows from the fact that the maximal quantum value of these party-lifted Bell-like inequalities is the same as the original Bell-like inequality [cf. Eq.~\eqref{Eq:size-k}], and is already attainable using a quantum state of entanglement depth $n$.

\section{Concluding Remarks}\label{con}

Lifting,  as introduced by Pironio \cite{PIR05}, is a procedure that allows one to systematically construct Bell inequalities for {\em all}  Bell scenarios starting from a Bell inequality applicable to a simpler scenario. It is known that Pironio's lifting preserves the facet-defining property of Bell inequalities, and thus lifted Bell inequalities (in particular, those lifted from the CHSH Bell inequality) can be found in {\em all} nontrivial Bell scenarios. In this work, we show that lifting leaves the maximal quantum and non-signaling value of Bell inequalities unchanged. 

Naturally, one may  ask whether the quantum state and local measurements maximally violating a lifted Bell inequality are  related to that of the original Bell inequality. Indeed, we show that Pironio's lifting also preserves the self-testability of a quantum state. Hence, the quantum state maximally violating a lifted Bell inequality is---modulo irrelevant local degrees of freedom---the same as that of the original inequality. Likewise, the self-testability of given local measurements is preserved using {\em any} but the outcome-lifting procedure.

The maximizers of lifted Bell inequalities are, as we show, generally not unique. Consequently, it is impossible to use the observed quantum value of such an inequality to self-test both the underlying state and {\em all} the local measurements: in the case of an input-lifted Bell inequality, no conclusions can be drawn regarding the additional measurements that do not appear in the inequality; in the case of a party-lifted Bell inequality, nothing can be said about the measurements of the additional party; in the case of an output-lifted Bell inequality, the self-testing of all the local POVM elements is impossible, but the self-testing of their combined effect seems possible. In fact, our numerical results (see Appendix~\ref{App:Self-test}) suggest that such a self-testing is just as robust as the original Bell inequality. Thus, Bell inequalities lifted from CHSH serve as generic examples whose maximal quantum violation  can be used to self-test a state, but not the underlying measurements in its entirety. 

Notice also that the non-uniqueness mentioned above evidently becomes more and more pronounced as the number of ``irrelevant" degrees of freedom increases, for example, by repeatedly applying lifting to a given Bell inequality. Since {\em only} correlation $\vec{P}$ belonging to the boundary of $\Q$ could violate a linear Bell inequality maximally, as the complexity of the Bell scenario (say, in terms of the number of measurements, outcomes, or parties) increases, it is conceivable that one can  always find a flat boundary of $\Q$ (corresponding to those of the lifted Bell inequality) with increasing dimension. Proving this statement rigorously and finding the exact scaling of the dimension of these flat boundaries would be an interesting direction to pursue for future research.

Besides, it will be interesting to see---in comparison with the original Bell inequality---whether the robustness in self-testing that we have observed for a particular version of the outcome-lifted CHSH Bell inequality is generic. From our example for the outcome-lifted CHSH inequality, it becomes clear that self-testing of the combined POVM elements is (sometimes) possible even if the self-testing of all individual POVM elements is not. This possibility opens another direction of research in the context of self-testing. In addition, our results also prompted the following question: does there exist a physical situation (say, the observation of the maximal quantum value of some Bell inequality) where the underlying measurements can be self-tested, but not the underlying state? Since the self-testing of measurements is seemingly more demanding than that for a quantum state, it is conceivable that no such examples can be constructed. Proving that this is indeed the case, however, clearly lies beyond the scope of the present paper.

{\em Note added:} While completing this manuscript, we became aware of the work of~\cite{Rosset2019private} which also discusses, among others, extensively the properties of liftings, as well as the work of~\cite{JedPrivate}, which exhibits examples of quantum correlations that can only be used to self-test the measured quantum state but not the underlying measurements.

\begin{acknowledgements}
We thank Jean-Daniel Bancal, J{\k{e}}drzej Kaniewski, Denis Rosset, and Ivan \v{S}upi\'c for useful discussions.  This work is supported by the Foundation for the Advancement of Outstanding Scholarship, Taiwan as well as the Ministry of Science and Technology, Taiwan (Grants No. 104-2112-M-006-021-MY3, No. 107-2112-M-006-005-MY2,  107-2627-E-006-001, 108-2811-M-006-501, and 108-2811-M-006-515).
\end{acknowledgements}

\appendix

\section{Examples of $\vecP\in\Q$ violating lifted Bell inequalities maximally }
\label{Examples}

To illustrate the non-unique nature of the maximizers of lifted Bell inequalities, consider the CHSH Bell inequality~\cite{Clauser69}:
\begin{align}\label{chshineq}
	\sum_{x,y,a,b=0,1} (-1)^{xy+a+b} P(ab|xy)\stackrel{\L}{\le} 2
\end{align}
as our seed inequality. Since this is a Bell inequality defined in the simplest Bell scenario (with two binary-output per party), its liftings can be found in all nontrivial Bell scenarios.

The quantum bound and nonsignaling bound of the above Bell inequality
are given, respectively, by  $\beta_{\Q}=2\sqrt{2}$ and $\beta_{\N}=4$. 
It is known~\cite{PR92} that the maximal quantum violation of the CHSH inequality
can be used to self-test (up to local isometry) the two-qubit maximally 
entangled state $\ket{\phi^+}=\left(\ket{00}+\ket{11}\right)/\sqrt{2}$ and 
the Pauli observables $\{\sigma_z,\sigma_x\}$ 
on one side and the Pauli observables $\{(\sigma_x+\sigma_z)/\sqrt{2},(\sigma_x-\sigma_z)/\sqrt{2}\}$ on the other. 
Thus, the correlation that gives the quantum maximum of the  inequality~\eqref{chshineq} is  {\em unique} 
and is given by
\be\label{Eq:TsirelsonPoint}
P_Q(ab|xy)=\tfrac{1}{4}+(-1)^{a+ b + x y}\tfrac{\sqrt{2}}{8}, \,\,\,  a,b,x,y\in\{0,1\},
\ee
where the $+1$-outcome of the observables is identified with the 0-th outcome in the conditional outcome probability distributions.

\subsection{Lifting of Inputs}
\label{App:LI}

For input-lifting, consider now a bipartite Bell scenario where Bob has instead three binary inputs  $(y=0,1,2)$. In this new Bell scenario, the following Bell inequality:
\begin{align}\label{liftinpchshineq}
	I^\text{\tiny LI-CHSH}_2:=\sum_{x,y,a,b=0,1} (-1)^{x y + a + b}  \tilde{P}(ab|xy)  \stackrel{\L}{\le} 2,
\end{align}
with $\{\tilde{P}(ab|xy)\}_{a,b,x=0,1,y=0,1,2}$ can be obtained by applying input-lifting to inequality~\eqref{chshineq}.

To illustrate the non-uniqueness of its maximizers, one may employ, e.g.,  either of the two trivial measurements for the third measurement $(y=2)$:
\begin{equation}
	M^{(2)}_{b|2}=\one\delta_{b,0}\quad\text{or}\quad M^{(2)}_{b|2}=\one\delta_{b,1}.
\end{equation}
Correspondingly, one obtains, in addition to [cf. Eq.~\eqref{Eq:TsirelsonPoint}] 
\begin{subequations}\label{Eq:PTilde}
\begin{equation}
	\tilde{P}^\text{\tiny LI}_{Q_1}(ab|x,y)=\tilde{P}^\text{\tiny LI}_{Q_2}(ab|x,y)=P_Q(a,b|x,y)
\end{equation}
for $a,b,x,y=0,1$, the distributions
\begin{equation}
	\tilde{P}^\text{\tiny LI}_{Q_1}(ab|xy)=\frac{1}{2}\delta_{b,0}\quad\text{and}\quad \tilde{P}^\text{\tiny LI}_{Q_2}(ab|x,y)=\frac{1}{2}\delta_{b,1},
\end{equation}
\end{subequations}
for $a,b,x=0,1$ but $y=2$.\footnote{More abstractly, these correlations can also be obtained from Eq.~\eqref{Eq:TsirelsonPoint} by 
applying input operations as given by Eq. (10) in \cite{Jul14}.} 

It is then easily verified that both these correlations violate inequality~\eqref{liftinpchshineq} to its quantum maximum of $2\sqrt{2}$. In fact, since the Bell expression of Eq.~\eqref{liftinpchshineq} is linear in $\tilde{P}$, it follows that an arbitrary convex combination of $\tilde{P}^\text{\tiny LI}_{Q_1}$ and $\tilde{P}^\text{\tiny LI}_{Q_2}$
\be\label{QBLCHSH}
	\tilde{P}(ab|xy)=p\tilde{P}^\text{\tiny LI}_{Q_1}(ab|xy) +(1-p)\tilde{P}^\text{\tiny LI}_{Q_2}(ab|xy),
\ee
where $0\le p \le 1$, must also give the maximal quantum value of Bell inequality~\eqref{liftinpchshineq}. Moreover, from the convexity of the set of quantum correlations $\Q$, we know that an arbitrary convex combination of the two correlations given above is also quantum realizable. Geometrically, this means that the set of $\tilde{P}\in\Q$ defined by Eq.~\eqref{QBLCHSH} forms a one-dimensional flat region of the quantum boundary. In Fig.~\ref{Fig1}, we show a two-dimensional slice on the space of correlations spanned by $\tilde{P}^\text{\tiny LO}_{Q_1}$, $\tilde{P}^\text{\tiny LO}_{Q_2}$, and the uniform distribution $\tilde{P}_0$. Note that on this peculiar slice, even $\Q$ appears to be a polytope.

\subsection{Lifting of Outcomes}
\label{App:LO}

For output-lifting, consider a bipartite Bell scenario where Bob's measurements have instead three outcomes  $(b=0,1,2)$. In this  Bell scenario, the following Bell inequality:
\begin{align}\label{liftchshineq}
	I^\text{\tiny LO-CHSH}_2:=\sum_{x,y,a=0,1}\sum_{b=0,1,2}\!\!\! (-1)^{x y + a+b} P(ab|xy)  \stackrel{\L}{\le} 2\end{align}
can be obtained by applying outcome-lifting to the $b=0$ outcome of Bob's measurements in inequality~\eqref{chshineq}.
Here, $b=2$ is the new outcome.

Indeed, it is readily seen that the  two correlations:
\begin{gather}
\label{EQCL1}
P^\text{\tiny LO}_{Q_1}(ab|xy)  
= \left[\tfrac{1}{4}+(-1)^{a+ b + x y}\tfrac{\sqrt{2}}{8}\right](1-\delta_{b,2}),\\
\label{EQCL2}
P^\text{\tiny LO}_{Q_2}(ab|xy)  
=\left[\tfrac{1}{4}+(-1)^{a+ b + x y}\tfrac{\sqrt{2}}{8}\right](1-\delta_{b,0})
\end{gather}
as well as
\begin{equation}\label{Eq:OL34}
\begin{split}
	P^\text{\tiny LO}_{Q_3}(ab|xy)&=P^\text{\tiny LO}_{Q_1}(ab|xy)\delta_{y,0}+P^\text{\tiny LO}_{Q_2}(ab|xy)\delta_{y,1},\\
	P^\text{\tiny LO}_{Q_4}(ab|xy)&=P^\text{\tiny LO}_{Q_1}(ab|xy)\delta_{y,1}+P^\text{\tiny LO}_{Q_2}(ab|xy)\delta_{y,0}
\end{split}
\end{equation}
all give rise to the quantum maximum of $2\sqrt{2}$ for the Bell inequality of Eq.~\eqref{liftchshineq}.

To see that they are indeed quantum realizable, we note that the correlation given in Eq. (\ref{EQCL1}) can be produced by employing the quantum strategy used to produce the correlation given in Eq.~\eqref{Eq:TsirelsonPoint}. By construction, Bob's measurement only produces two outcomes labeled by $b=0$ and $b=1$, and thus the $b=2$ outcome never appears, as required in Eq.~\eqref{EQCL1}. 
To obtain the correlation given in Eq. (\ref{EQCL2}), one may start from the correlation of Eq.~\eqref{EQCL1} and apply the classical relabeling of $b=0\leftrightarrow b=2$. 

Similarly, the two correlations of Eq.~\eqref{Eq:OL34} can be realized by first implementing the quantum strategy that realizes $\vec{P}^\text{\tiny LO}_{Q_1}$, followed by applying the classical relabeling of $b=0\leftrightarrow b=2$ depending on whether $y=0$ or $y=1$ \footnote{This belongs to the class of outcome operation given by Eq. (9) in \cite{Jul14}.}. Since $\{\vec{P}^\text{\tiny LO}_{Q_i}\}_{i=1}^4$ forms a linearly independent set, and an arbitrary convex combination of them also gives the quantum bound, we thus see that the quantum face\footnote{The set of  quantum correlations that gives the quantum bound of a Bell inequality is called quantum face, see~\cite{GKW+18}.} of the outcome-lifted inequality of \eqref{liftchshineq} is (at least) three-dimensional.

\begin{figure}[t!]
\includegraphics[width=7.5cm]{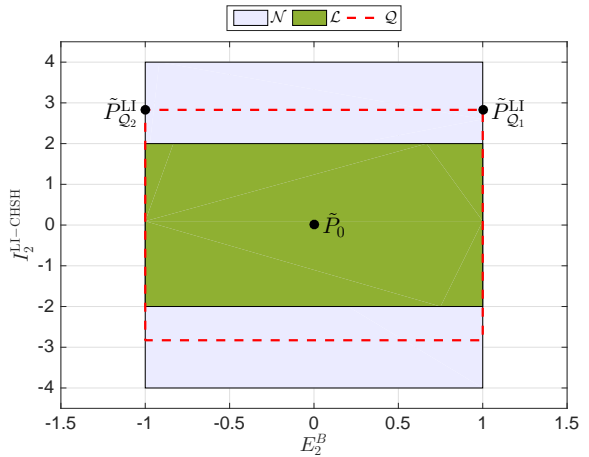}
\caption{A two-dimensional slice in the input-lifted space of correlations spanned by $\tilde{P}^\text{\tiny LO}_{Q_1}$, $\tilde{P}^\text{\tiny LO}_{Q_2}$ [see Eq.~\eqref{Eq:PTilde}] and the uniform distribution $\tilde{P}_0$. From the innermost to the outermost, we have, respectively, the set of Bell-local correlations $\L$ (green), the set of quantum correlations $\Q$ (red, dashed boundary), and the set of nonsignaling correlations $\N$ (mauve). To illustrate the degeneracy in the maximally-violating correlations,  we have chosen the input-lifted Bell inequality of Eq.~\eqref{liftinpchshineq} and the marginal correlator for $y=2$, i.e., $E^B_2=\tilde{P}(b=0|y=2)-\tilde{P}(b=1|y=2)$ as, respectively,  the vertical and horizontal axis of this plot.   As opposed to the two-dimensional slices shown in~\cite{GKW+18}, the set of quantum correlations $\Q$ appears to be a rectangle on this slice.
\label{Fig1}}
\end{figure}

\subsection{Lifting of Party}

Geometrically, party-lifting also introduces degeneracy in the maximizers of a Bell inequality. For example, a possible party-lifting of the CHSH Bell inequality of Eq.~\eqref{chshineq} to the three-party, two-input, two-output Bell scenario reads as:
\begin{equation}\label{Eq:CHSH:PartyLifted}
	\sum_{x,y,a,b=0,1} \left\{\left[(-1)^{xy+a+b}-\tfrac{1}{2}\right] P(ab0|xy0)\right\}\stackrel{\L}{\le} 0
\end{equation}

As mentioned in the proof of Corollary~\ref{Res:PartyLifting}, maximizers of this Bell inequality must be such that the bipartite marginal distribution $P(ab|xy)$ violates CHSH Bell inequality maximally while the marginal distribution for the third party must satisfy $P(0|0)=1$. However, these conditions do not impose any constraint on the other marginal distribution $P(c|z)$ for $z\neq 0$. In particular, both the choice of $P(c|1)=\delta_{c,0}$ and $P(c|1)=\delta_{c,1}$ would fulfill the above requirement. As such, although the quantum face of the CHSH Bell inequality is a point in the correlation space, the quantum face of the party-lifted Bell inequality of Eq.~\eqref{Eq:CHSH:PartyLifted} has become one-dimensional.

\section{Quantum realizability of distributions obtained by grouping and splitting outcomes}\label{QRpre}

In this Appendix, we provide the details showing how one can---while preserving between the quantum violation of a Bell inequality and its outcome-lifted version---realize quantum-mechanically a fewer-outcome (more-outcome) correlation if the original more-outcome (fewer-outcome) correlation is quantum.

\subsection{Grouping of outcomes}

 Suppose that the joint probabilities $P(ab'|xy')$ and $P(au|xy')$ on the right-hand-side of Eq. (\ref{coarse-grain}) are realized by
 a quantum state $\rho_{12}$ with the POVM $\{M^{(1)}_{a|x}\}_{a,x}$ on Alice's side and $\{M^{(2)}_{b|y}\}_{b,y}$ 
 on Bob's side, cf. Eq.~\eqref{Eq:Born2}. Then the joint probabilities $\tilde{P}(ab'|xy')$ appearing on the left-hand-side  of Eq. (\ref{coarse-grain}), which corresponds to an effective $v$-outcome distribution, is realizable  by the same quantum state $\rho_{12}$ with the same POVM $\{M^{(1)}_{a|x}\}_{a,x}$ on Alice's side 
 and the following POVM $\tilde{M}^{(2)}_{b|y}$ on Bob's side:\footnote{That $\tilde{M}^{(2)}_{b|y}$ satisfies both the positivity constraints and the normalization constraints is evident from Eq.~\eqref{Eq:POVM:Grouping}.}
 \begin{equation}\label{Eq:POVM:Grouping}
 \begin{split}
 	\tilde{M}^{(2)}_{b|y} &= M^{(2)}_{b|y}\quad y\neq y',\\
 	\tilde{M}^{(2)}_{b|y'} &= M^{(2)}_{b|y'}\quad b\not\in\{b',u\},\\
 	\tilde{M}^{(2)}_{b'|y'} &= M^{(2)}_{b'|y'}+ M^{(2)}_{u|y'}.
 \end{split}
 \end{equation}
With this choice, it follows from 
\begin{equation}\label{Eq:Born3}
	\tilde{P}(ab'|xy')=\tr (M^{(1)}_{a|x} \otimes  \tilde{M}^{(2)}_{b'|y'}  \rho_{12}),
\end{equation}
Eq.~\eqref{Eq:POVM:Grouping}, and 
 \begin{align}\label{Eq:BornRule:Combined}
 \tr (M^{(1)}_{a|x} \otimes \tilde{M}^{(2)}_{b'|y'} \rho_{12})&=\tr (M^{(1)}_{a|x} \otimes {M}^{(2)}_{b'|y'} \rho_{12}) \nonumber \\ 
 &+\tr ( M^{(1)}_{a|x} \otimes {M}^{(2)}_{u|y'} \rho_{12})
 \end{align}
 that Eq.~\eqref{coarse-grain} is satisfied, and hence that the violation-preserving fewer-outcome correlation is indeed attainable by coarse graining, i.e., the grouping of outcomes.
 
 \subsection{Splitting of outcomes}
 \label{App:Split}
 
On the other hand, if we instead start from the original (fewer-outcome) Bell scenario, then the joint probabilities appearing on the right-hand-side  of Eq.~\eqref{Eq:fine_grain}, which corresponds to a $v+1$-outcome distribution, can be realized by, e.g.,  employing the same quantum state $\rho_{12}$ with the same POVM $\{M^{(1)}_{a|x}\}_{a,x}$ on Alice's side and the following POVM $\tilde{M}^{(2)}_{b|y}$ on Bob's side:
 \begin{equation}\label{Eq:POVM:Splitting}
 \begin{split}
 	M^{(2)}_{b|y} &= \tilde{M}^{(2)}_{b|y}\quad y\neq y',\\
 	M^{(2)}_{b|y'} &= \tilde{M}^{(2)}_{b|y'}\quad b\not\in\{b',u\},\\
 	M^{(2)}_{b'|y'} = p\tilde{M}^{(2)}_{b'|y'},&\quad M^{(2)}_{u|y'} = (1-p)\tilde{M}^{(2)}_{b'|y'},
 \end{split}
 \end{equation}
for arbitrary $0\le p \le 1$. The positivity of the left-hand-side of Eq.~\eqref{Eq:POVM:Splitting} and their normalization are evident from their definition. Moreover, using 
Eq.~\eqref{Eq:BornRule:Combined} and
\begin{equation}\label{Eq:Born3}
	\widehat{P}(ab|xy)=\tr (M^{(1)}_{a|x} \otimes  {M}^{(2)}_{b'|y'}  \rho_{12}),
\end{equation}
it is easy to see that Eq.~\eqref{Eq:fine_grain} holds with the assignment given in Eq.~\eqref{Eq:POVM:Splitting}. Hence, the POVM given by the left-hand-side of Eq.~\eqref{Eq:POVM:Splitting} indeed realizes the required violation-preserving more-outcome correlation by splitting the $b'$-outcome of Bob's $y'$-th measurement.

\section{Robust self-testing based on the quantum violation of the outcome-lifted CHSH inequality}
\label{App:Self-test}

We show in this Appendix that self-testing via the quantum violation of the outcome-lifted inequality of Eq.~\eqref{liftchshineq} is robust. In this regard, note that the maximal quantum violation of inequality~\eqref{liftchshineq} can also be achieved by Alice and Bob sharing the following two-qubit maximally entangled state
\begin{equation}\label{Eq:TargetState}
	\ket{\tilde{\psi}} = \cos\frac{\pi}{8}\frac{\ket{00}-\ket{11}}{\sqrt{2}} + \sin\frac{\pi}{8}\frac{\ket{01}+\ket{10}}{\sqrt{2}},
\end{equation}
while performing the optimal qubit measurements for Alice:
\begin{equation}\label{Eq:POVMA}
\begin{aligned}
\tilde{M}_{0|0}^{(1)} = \frac{1}{2}(\openone + \sigma_z),\quad \tilde{M}_{1|0}^{(1)} = \frac{1}{2}(\openone - \sigma_z),\\
\tilde{M}_{0|1}^{(1)} = \frac{1}{2}(\openone + \sigma_x),\quad \tilde{M}_{1|1}^{(1)} = \frac{1}{2}(\openone - \sigma_x),\\
\end{aligned}
\end{equation}
and for Bob:
\begin{equation}\label{Eq:POVMB}
\begin{aligned}
\tilde{M}_{0|0}^{(2)} = E_{0|0},\quad \tilde{M}_{1|0}^{(2)} = \frac{1}{2}(\openone - \sigma_z),\quad \tilde{M}_{2|0}^{(2)} = E_{2|0},\\
\tilde{M}_{0|1}^{(2)} = E_{0|1},\quad \tilde{M}_{1|1}^{(2)} = \frac{1}{2}(\openone - \sigma_x),\quad \tilde{M}_{2|1}^{(2)} = E_{2|1},\\
\end{aligned}
\end{equation}
where $\{E_{b|y}\}_{b=0,2}$ are any valid POVM elements satisfying $\sum_{b=0,2} E_{b|y} = \openone - M_{1|y}^{(2)}$ for all $y$. 

Notice that Eq.~\eqref{Eq:StateTransformation} and Eq.~\eqref{Eq:POVMTransformation} only hold for the case of perfect self-testing of the reference state and reference measurements. To demonstrate robust self-testing for the above reference state $\ket{\tilde{\psi}}$ and reference measurements $\{\tilde{M}^{(1)}_{a|x}\}_{a,x}$, $\{\tilde{M}^{(2)}_{b|y}\}_{b,y}$, we follow the approach of~\cite{YVB+14} to arrive at statements saying that if the observed quantum violation of inequality~\eqref{liftchshineq} is close to its maximal value, then (1) the measured system contains some degrees of freedom that has a high fidelity with respect to the reference state $\ket{\tilde{\psi}}$, and (2) with high probability, the uncharacterized measurement devices function like $\{\tilde{M}^{(1)}_{a|x}\}_{a,x}$, $\{\tilde{M}^{(2)}_{b|y}\}_{b,y}$ acting on the same degrees of freedom.

\subsection{Robust self-testing of the reference state}

To this end, we shall make use of the swap method proposed in~\cite{YVB+14}. The key idea is to introduce local swap operators $\Phi_1,\Phi_2$ so that the state acting on Alice's and Bob's Hilbert space (of unknown dimension) gets swapped locally with some auxiliary states of trusted Hilbert space dimension (qubit in our case). To better understand how this works, let us first consider an  example with characterized devices before proceeding to the case where the devices are uncharacterized.

For this purpose, let us concatenate the following controlled-not (CNOT) gates
\begin{equation}\label{Eq:UV}
\begin{aligned}
&U_1 = \openone\otimes \proj{0} + \sigma_x\otimes\proj{1},\\
&V_1 = \proj{0}\otimes\openone + \proj{1}\otimes\sigma_x,
\end{aligned}
\end{equation}
to obtain the (two-qubit) swap gate $\Phi_1 = U_1V_1U_1$ acting on $\mathcal{H}_A\otimes\mathcal{H}_{A'}$ (see Sec.~\ref{Sec:SelfTest}). We may define a swap operator $\Phi_2=U_2V_2U_2$ acting on Bob's systems in exactly the same way. Importantly, one notices from Eq.~\eqref{Eq:POVMA} and Eq.~\eqref{Eq:POVMB} that it is possible to express the individual unitaries in terms of the POVM elements leading to the maximal quantum violation of inequality~\eqref{liftchshineq}. For example, one may take
\begin{equation}
	\sigma_z = \openone - 2 \tilde{M}_{1|0}^{(i)},\quad \sigma_x = \openone - 2 \tilde{M}_{1|1}^{(i)}\quad\forall\,\, i=1,2.
\end{equation}

Moreover, if we define the global swap gate by $\Phi=\Phi_1\otimes\Phi_2$ and denote the state acting on $\mathcal{H}_A\otimes\mathcal{H}_B$ by $\rho_{12}$, then the ``swapped'' state is:
\begin{equation}\label{Eq:SwappedState}
	\rho^{\text{\tiny SWAP}}:=\tr_{12}[\Phi(\proj{0}\otimes \rho_{12} \otimes\proj{0})\Phi^\dag]
\end{equation}
where $\tr_{12}$ represents a partial trace over the Hilbert space of $\mathcal{H}_A\otimes\mathcal{H}_B$. When $\Phi$ is exactly the swap gate defined above via Eq.~\eqref{Eq:UV}, $\rho^{\text{\tiny SWAP}}$ is exactly $\rho_{12}$. Thus, the fidelity between $\rho^{\text{\tiny SWAP}}$ and the reference state $\ket{\tilde{\psi}}$: $F=\langle \tilde{\psi}|\rho^{\text{\tiny SWAP}}|\tilde{\psi}\rangle$ provides a  figure of merit on the similarity between (some relevant parts of) the shared state $\rho_{12}$ and the reference state $\ket{\tilde{\psi}}$.

To perform a device-independent characterization, the assumption of $\tilde{M}_{1|x}^{(1)}$ and $\tilde{M}_{1|y}^{(2)}$ is relaxed to unknown projectors $M_{1|x}^{(1)}$ and $M_{1|y}^{(2)}$ (acting on Hilbert space of arbitrary dimensions), and the corresponding ``CNOT'' gates becomes
\begin{equation}
\begin{aligned}
&U_i = \openone\otimes \proj{0} + (\openone - 2 M_{1|1}^{(i)})\otimes\proj{1},\\
&V_i = (\openone - M_{1|0}^{(i)}) \otimes\openone + M_{1|0}^{(i)}\otimes\sigma_x,
\end{aligned}
\label{EqApp:_CNOTs_DI}
\end{equation}
for $i=1,2$. One can verify that the fidelity $F=\langle \tilde{\psi}|\rho^{\text{\tiny SWAP}}|\tilde{\psi}\rangle$ is then a linear function of the moments such as $\langle M_{a|x}^{(1)}\otimes M_{b|y}^{(2)}\rangle$, $\langle M_{a|x}^{(1)}\otimes M_{b|y}^{(2)}M_{b'|y'}^{(2)}\rangle$ etc.,  where $\langle\cdot\rangle:=\tr(\cdot\rho_{12})$. 

Thus, a lower bound on $F$ for any observed value of Bell inequality violation (without assuming the shared state or the measurements performed) can be obtained by solving the following semidefinite program:
\begin{equation}
\begin{aligned}
\min\quad &F\\
\text{such that} \quad& \Gamma^S\succeq 0,\\
&I_2^{\text{\tiny LO-CHSH}} =  I_{2,\text{\tiny obs}}^{\text{\tiny LO-CHSH}},
\end{aligned}
\label{EqApp:min_fidelity_SDP}
\end{equation}
where $\Gamma^S$ is any Navascu{\'e}s-Pironio-Ac{\'i}n-type~\cite{NPA07} moment matrix that contains all the moments appearing in $F$. In our computation, we employed a moment matrix that is built from a sequence of operators $S$ that contains all operators from level 1+AB (or equivalently, level 1 from the hierarchy of Ref.~\cite{MBL+13}) and some additional operators from level 3. 

Our results (see Fig.~\ref{fig_min_fidelity_LOCHSH}) clearly shows that the self-testing property of $I_2^{\text{\tiny LO-CHSH}}$ with respect to the reference maximally entangled state $\ket{\tilde{\psi}}$ of Eq.~\eqref{Eq:TargetState} is indeed robust. In other words, as long as the observed violation of $I_2^{\text{\tiny LO-CHSH}}$ is greater than $\approx2.4$, one can still obtain a non-trivial lower bound on the fidelity ($>1/2$) with respect to $\ket{\tilde{\psi}}$. Moreover, a separate computation using the original CHSH Bell inequality of Eq.~\eqref{chshineq} (and the same level of approximation of $\Q$) gives---within the numerical precision of the solver---the same curve, thereby suggesting that the outcome-lifted Bell inequality of~\eqref{liftchshineq} offers the same level of robustness as compared with its seed inequality.

\begin{figure}
\includegraphics[width=0.9\linewidth]{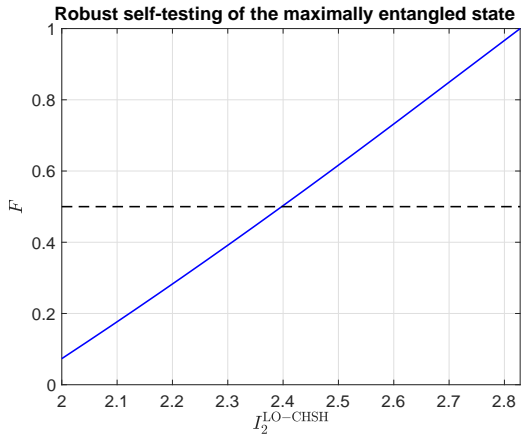}
\caption{Lower bounds on the fidelity as a function of the value of the outcome-lifted CHSH inequality $I_2^{\text{\tiny LO-CHSH}}$. The results are obtained by solving the semidefinite program described in Eq.~\eqref{EqApp:min_fidelity_SDP}.
}
\label{fig_min_fidelity_LOCHSH}
\end{figure}

\subsection{Robust self-testing of Alice's POVM}

Even though it is impossible to completely self-test all local measurements, robust self-testing of Alice's POVM---as one would intuitively expect---can still be achieved. In particular, when the observed violation of $I_2^{\text{\tiny LO-CHSH}}$ is close to the quantum bound of $2\sqrt{2}$, it must be the case that Alice's measurements (on the relevant degrees of freedom) indeed behave like measurements in the $\sigma_z$ and $\sigma_x$ basis, respectively, for $x=0,1$. 

To that end, we again make use the swap method proposed in Ref.~\cite{YVB+14}. The idea is that if these measurements behave as expected, then their measurements on the auxiliary states swapped into the uncharacterized device, i.e., $\Phi_1(|\varphi\rangle)$---with $\ket{\varphi}$ being eigenstates of $\sigma_z$ and $\sigma_x$---should produce outcomes $a$ with statistics  $\{P(a|x,|\varphi\rangle)\}$ satisfying
\begin{equation}
P(0|0,|0\rangle) = P(1|0,|1\rangle) = P(0|1,|+\rangle) = P(1|1,|-\rangle) = 1.
\label{EqApp:Pax_ideal}
\end{equation}

\begin{figure}
\includegraphics[width=0.9\linewidth]{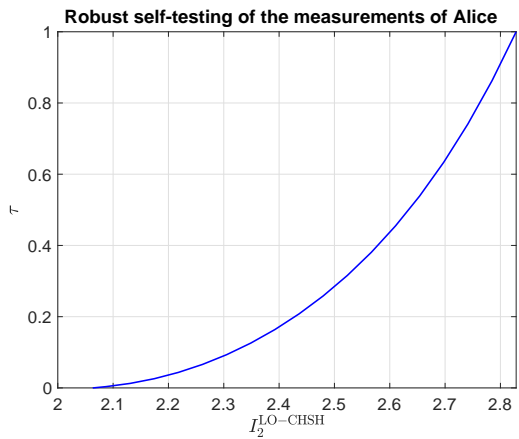}
\caption{Lower bounds on the figure of merit defined in Eq.~\eqref{EqApp:tau} as a function of the value of the outcome-lifted CHSH inequality $I_2^{\text{\tiny LO-CHSH}}$. The bounds are obtained by solving the semidefinite program described in Eq.~\eqref{EqApp:min_tau_SDP}.
}
\label{fig_min_tau_LOCHSH}
\end{figure}

Using the same swap operator defined via Eq.~\eqref{EqApp:_CNOTs_DI}, we get
\begin{equation}\label{Eq:ProbOutcome}
	P(a|x,|\varphi\rangle) = \tr\left\{M_{a|x}^{(1)}\left[\Phi_1\big( \rho_{12}\otimes\proj{\varphi}\big) \Phi_1^\dag\right]\right\},
\end{equation}
where we have, for simplicity, omitted the identity operator  acting on Bob's system. Notice that the left-hand-side of Eq.~\eqref{Eq:ProbOutcome} is again some linear combination of moments. Likewise for the following figure of merit~\cite{YVB+14}:
\begin{equation}
\begin{aligned}
\tau = \frac{1}{2}\big[  &P(0|0,|0\rangle) + P(1|0,|1\rangle)\\
+ &P(0|1,|+\rangle) + P(1|1,|-\rangle)   \big] - 1,
\end{aligned}
\label{EqApp:tau}
\end{equation}
which takes value between $-1$ and $+1$. The maximum of $+1$, in particular, happens {\em only} when Alice's POVM $M^{(1)}_{a|x}$ correspond to measurements in $\sigma_z$ and $\sigma_x$, respectively, for $x=0,1$. $\tau$ therefore quantifies the extent to which the measurement devices function like the reference measurements. Given a violation of $I_2^{\text{\tiny LO-CHSH}}$, a lower bound on $\tau$ can thus be obtained by solving the following semidefinite program:
\begin{equation}
\begin{aligned}
\min\quad &\tau\\
\text{such that} \quad& \Gamma^S\succeq 0,\\
&I_2^{\text{\tiny LO-CHSH}} =  I_{2,\text{\tiny obs}}^{\text{\tiny LO-CHSH}}.
\end{aligned}
\label{EqApp:min_tau_SDP}
\end{equation}
The resulting lower bounds on $\tau$ are shown in Figure~\ref{fig_min_tau_LOCHSH}. We see that the value of $1$ is obtained when the maximal quantum value of $I_2^{\text{\tiny LO-CHSH}}$ is observed, which means that the reference measurements of Eq.~\eqref{Eq:POVMA} are correctly certified in this case. For nonmaximal values of $I_2^{\text{\tiny LO-CHSH}}$, we see that $\tau$ decreases accordingly. Importantly, as pointed out in Ref.~\cite{YVB+14}, the procedure of sending the prepared eigenstates into the swap gate is a virtual process that allows us to interpret the figure of merit operationally, but the result still holds without any assumption on the devices of interest.

\subsection{Partial but robust self-testing of Bob's POVMs}

Finally, we would like to show that the outcome-lifted CHSH inequality of Eq.~\eqref{liftchshineq} can also be used for a ``partial'' self-testing of Bob's optimal measurements. The steps are the same as those described in the self-testing of Alice's measurements. That is, the eigenstates of $\sigma_z$ and $\sigma_x$ are sent to the swap gate before Bob performs his  measurements $\{M_{b|y}^{(2)}\}$. 

To this end, we define the analog of Eq.~\eqref{EqApp:tau} as:
\begin{equation}
\begin{aligned}
\tau_3 = \frac{1}{2}\big[  &P(0|0,|0\rangle) + P(1|0,|1\rangle) + P(2|0,|0\rangle)\\
+ &P(0|1,|+\rangle) + P(1|1,|-\rangle) + P(2|1,|+\rangle)  \big] - 1,
\end{aligned}
\label{EqApp:tau3}
\end{equation}
and introduce a further figure of merit
\begin{equation}
\tau_1 =  P(1|0,|0\rangle) + P(1|1,|+\rangle) -1
\label{EqApp:tau1}
\end{equation}
to self-test only the POVM element corresponding to Bob's outcome 1 for both measurements.

Thus, $\tau_3$ takes into account of Bob's all measurements outcomes while $\tau_1$ only involves the second measurement outcome. All these figures of merit range from $-1$ to $+1$, and $+1$ is recovered for 
\begin{itemize}
\item $\tau_3$ if Bob's measurement device acts on the swapped eigenstate according to Eq.~\eqref{Eq:POVMB};
\item $\tau_1$ if Bob's measurement device acts on the swapped eigenstate in such a way that the 2nd POVM element for each measurement functions according to Eq.~\eqref{Eq:POVMB}
\end{itemize}
In other words, the value of $\tau_1$ measures the extent to which $M_{1|y}^{(2)}$ behaves according to that prescribed in Eq.~\eqref{Eq:POVMB}, while the value of $\tau_3$ further indicates if the combined effect of $M_{0|y}^{(2)}+M_{2|y}^{(2)}$ 
also behaves according to that prescribed in Eq.~\eqref{Eq:POVMB}.

By solving the semidefinite program of Eq.~\eqref{EqApp:min_tau_SDP} using the appropriate objective functions, we obtained lower bounds on each figure of merit as a function of the quantum violation of $I_2^{\text{\tiny LO-CHSH}}$. As shown in Fig.~\ref{fig_min_tau_Bob_LOCHSH}, the bounds on $\tau_3$ and $\tau_1$ when $I_2^{\text{\tiny LO-CHSH}}$ takes its maximal value successfully self-tests, respectively,  the combined effect of $M_{0|y}^{(2)}+M_{2|y}^{(2)}$ as well as that of $M_{1|y}^{(2)}$. In summary, for the outcome-lifted CHSH inequality of Eq.~\eqref{liftchshineq}, where the first outcome is lifted in each of Bob's measurement, it is still possible to self-test Alice's optimal measurements and the overall behavior of Bob's measurements.

\begin{figure}
\includegraphics[width=0.9\linewidth]{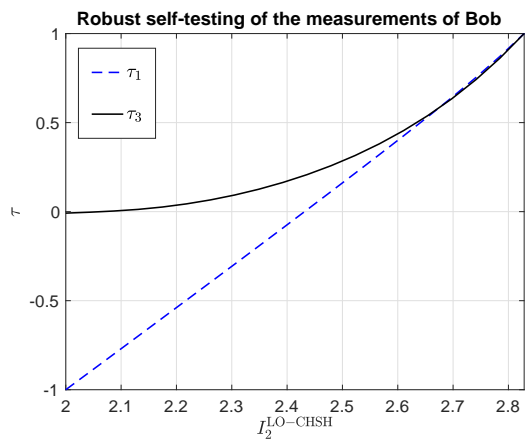}
\caption{
Lower bounds on the figures of merit defined in Eqs.~\eqref{EqApp:tau3}-\eqref{EqApp:tau1}  as a function of the value of the outcome-lifted CHSH inequality $I_2^{\text{\tiny LO-CHSH}}$. The bounds are obtained by solving the semidefinite program described in Eq.~\eqref{EqApp:min_tau_SDP} with the appropriate figure of merit. The two figures of merit, as explained in the text, reflect different aspects of the self-testability of Bob's measurements.
}
\label{fig_min_tau_Bob_LOCHSH}
\end{figure}

\bibliography{JHL} 

\end{document}